\RequirePackage{amssymb}
\documentclass{article}
\usepackage{mathtools,amsthm}
\usepackage{color,svgcolor}
\usepackage{graphicx}
\usepackage{multirow}
\usepackage{amsfonts}
\usepackage{enumitem}
\usepackage{xspace}

\usepackage{algorithm}
\usepackage{algcompatible}
\newcommand{\To}{\textbf{to}\xspace}

\newcommand{\algorithmicreturn}{\textbf{return}}
\newcommand{\RETURN}{\STATE\algorithmicreturn{}\xspace}
\algnewcommand{\IfThen}[2]{%
  \State \algorithmicif\ #1\ \algorithmicthen\ #2}
\algnewcommand{\IfThenEnd}[2]{%
  \State \algorithmicif\ #1\ \algorithmicthen\ #2\ \algorithmicend\ \algorithmicif}
\algnewcommand{\IfThenElse}[3]{%
  \State \algorithmicif\ #1\ \algorithmicthen\ #2\ \algorithmicelse\ #3}
\algnewcommand{\ForDoEnd}[2]{%
  \State \algorithmicfor\ #1\ \algorithmicdo\ #2\ \algorithmicend\ \algorithmicfor}

\algnewcommand\algorithmicreadonly{\textbf{Read-only:}}
\algnewcommand\READONLY{\item[\algorithmicreadonly]}%

\newcommand{\F}{\ensuremath{\mathbb{F}}\xspace}
\newcommand{\SF}{\ensuremath{\mathbb{S}}\xspace}

\newcommand{\MatrixRank}[1]{\ensuremath{{\textup{Rank}\,{#1}}}}
\DeclareMathOperator{\pe}{\,\ensuremath{\mathrel{+}=}\,}

\usepackage{listings}
\lstdefinestyle{slp}{
   basicstyle=\ttfamily\scriptsize,
   breaklines=true,
   breakatwhitespace=true,
   breakindent=0pt,
   string=[d]{;},
   keywordstyle=\color{blue},
   showstringspaces=false,
   numbers=none
}

\newcommand{\plinopt}{\href{https://github.com/jgdumas/plinopt}{\textsc{PLinOpt}}\xspace}
\newcommand{\plinoptdata}[1]{\href{https://github.com/jgdumas/plinopt/tree/main/data}{\plinopt/data/\url{#1}}}

\newcommand{\LRP}{\textsc{lrp}\xspace}
\newcommand{\SLP}{\textsc{slp}\xspace}

\newcommand{\SSLP}{\mathrm{SSLP}}

\newcommand{\matrixsize}[2]{\ensuremath{{{#1}{\times}{#2}}}}

\newcommand{\Matrices}[2]{\ensuremath{{{#1}^{#2}}}}

\DeclareMathOperator{\tensorproduct}{\ensuremath{\otimes}}

\newcommand{\mat}[1]{\textbf{\sc{\ensuremath{#1}}}}
\newcommand{\Transpose}[1]{{{\mat{#1}}^{\intercal}}\xspace}

\usepackage{nicematrix}
\NiceMatrixOptions{%
    renew-dots,%
    xdots/shorten=.5em, %
    xdots/inter=.65em, %
}
\newenvironment{smatrix}[1][]{\left\lbrack\begin{NiceMatrix}[small,#1]}{\end{NiceMatrix}\right\rbrack}
\newcommand{\emailto}[1]{\href{mailto:#1}{#1}}
\newcommand{\affiliation}[1]{\footnote{#1}}
\newcommand{\grantsponsor}[3]{\href{#3}{#2}}
\newcommand{\grantnum}[2]{#2}
\title{Computational Explorations on Semifields}
\author{Jean-Guillaume Dumas%
\affiliation{
  {Universit\'e Grenoble Alpes, Laboratoire Jean Kuntzmann},
  {umr CNRS 5224, F38058 Grenoble},
  {France}.
\emailto{Jean-Guillaume.Dumas@univ-grenoble-alpes.fr}}
\and{Stefano Lia}%
\affiliation{
  {Ume\aa~University, Matematik och matematisk statistik},
  {901 87 Ume\aa},
  {Sweden}.
\emailto{Stefano.Lia@umu.se}
}
\and{John Sheekey}%
\affiliation{%
  {University College Dublin, School of Mathematics and Statistics},
  {Belfield,Dublin 4},
  {Ireland}.
\emailto{John.Sheekey@ucd.ie}
}
}

\usepackage{fullpage}
\usepackage{booktabs}
\usepackage{hyperref}

\makeatletter
\hypersetup{%
	pdftitle={\@title},
	pdfauthor={Jean-Guillaume Dumas and Stefano Lia and John Sheekey}
	plainpages=true,
	breaklinks=true,
	colorlinks=true,
	linkcolor=darkred,
	citecolor=blue,
	urlcolor=darkgreen,
	hypertexnames=false,
}
\makeatother

\usepackage{aliascnt}
\usepackage[capitalize]{cleveref}
\AddToHook{cmd/appendix/before}{%
    \crefalias{section}{appendix}%
    \crefalias{subsection}{appendix}
}

\let\oldtheorem\newtheorem
\RenewDocumentCommand{\newtheorem}{s m o m O{}}{%
\IfBooleanTF{#1}%
{\oldtheorem{#2}{#4}}%
{\IfNoValueTF{#3}{\oldtheorem{#2}{#4}[#5]}%
{\newaliascnt{#2}{#3}%
\oldtheorem{#2}[#2]{#4}%
\aliascntresetthe{#2}}}}

\newcommand{\algabbrev}{Alg}
\newcommand{\algname}{{\algabbrev}.}
\newcommand{\algsname}{{\algabbrev}s.}
\newcommand{\Algname}{Algorithm}
\newcommand{\Algsname}{{\Algname}s}
\crefname{algorithm}{\algname}{\algsname}
\Crefname{algorithm}{\Algname}{\Algsname}

\newtheorem{theorem}{Theorem}
\newtheorem{corollary}[theorem]{Corollary}

\newtheorem{definition}[theorem]{Definition}

\newtheorem{lemma}[theorem]{Lemma}

\newtheorem{remark}[theorem]{Remark}
\newtheorem{example}[theorem]{Example}

\begin{document}
\maketitle
\begin{abstract}
A finite semifield is a division algebra over a finite field where
multiplication is not necessarily associative.
We consider here the complexity of the multiplication
in small semifields and finite field extensions.
For this operation, the number of required base field multiplications
is the tensor rank, or the multiplicative complexity.
The other base field operations are additions and scalings by
constants, which together we refer to as the additive complexity.
When used recursively, the tensor rank determines the exponent while
the other operations determine the constant of the associated
asymptotic complexity bounds.
For small extensions, both measures are of similar importance.

In this paper, we establish the tensor rank of some semifields and
finite fields of characteristics 2 and 3.
We also propose new upper and lower bounds on their additive
complexity, and give new associated algorithms improving on the
state-of-the-art in terms of overall complexity. We achieve this by
considering short straight line programs for encoding linear codes
with given parameters.
\end{abstract}
\thanks{This work was partially supported by the
  \grantsponsor{ulysses}{PHC ULYSSES program (Campus France / Taighde
    \'Eireann -- Research Ireland)}{https://www.campusfrance.org/fr/ulysses}
  under grant
  \grantnum{ulysses}{\href{https://research.ie/assets/uploads/2023/09/Ulysses-2024-awardees.pdf}{50216SG}}
  as well as by the
  \grantsponsor{anr}{Agence Nationale de la Recherche}{https://anr.fr}
  under grant
  \grantnum{anr}{\href{https://anr.fr/ProjetIA-22-PECY-0010}{ANR-22-PECY-0010}}.}
\section{Introduction}\label{sec:intro}
Establishing the computational complexity of performing multiplication
in an algebra over a finite field is a well-studied problem. Famous
examples include matrix multiplication, polynomial multiplication, and
finite field extension multiplication. More recently, nonassociative
division algebras ({\em semifields}) have been addressed,
with applications to coding theory~\cite{Combarro:2011:IJCM},
cryptology~\cite{Rua:2018:sboxes} or graph theory~\cite{May:2007:issac}.

Most frequently, the multiplicative complexity ({\em tensor rank}) has
been the focus of study, in particular for recursive
implementations. Moreover, for small finite field polynomial
multiplications for example, the potentially fastest implementations
would exhaustively pre-compute the results in tables for all possible
polynomials.

In embedded cryptography, however, at least
since~\cite{Kocher:1996:TimingAttacks},
one way to ensure side-channel resistance is to avoid
branching and table lookups, as they can leak significant information
about the secrets (for instance, by timing or cache attacks, see,
e.g.~\cite{Kasper:2012:FastECCopenSSL,Bernstein:2012:High-speed}).
There, so-called constant-time implementations (i.e. not
depending on the actual values of same-size input) often fall back to
polynomial arithmetic. Therefore, it is important to minimize the total
number of field operations.
Another area of interest includes cryptography in larger finite fields,
implemented as towers of extension. There, the base field operations
can be implemented with table lookups, but the extensions also use
modular polynomial operations directly. It is again important to reduce
both the number of multiplications and additions to obtain fast
implementations; see e.g.~\cite{EGI:2011:africacrypt}.

By interpolation (the Toom-Cook method), when $q\geq 2n-1$, the tensor
rank of the multiplication in the field $\F_{q^n}$ is upper bounded by
$2n-1$~\cite{toom:1963:complexity}.
When $q$ is small, the question remains open in general.
The algorithms of~\cite{Barbulescu:2012:bilmaps,Covanov:2019:formulas}
can then find optimal formulae (straight line programs, \SLP) to
evaluate bilinear maps over finite fields.
Now, in \cite{Lavrauw:2022:tranksemi} it was shown that there exist
semifields with lower multiplicative complexity than the field of the
same order.
However, since semifield multiplication cannot usually be applied
recursively, it remains to consider the overall complexity.
The other base field operations are additions and scalings by
constants, which together we refer to as the additive complexity.
Up to our knowledge, much less is known about additive complexity,
even if optimal \SLP{s} are given for some Toom-Cook formulae in
\cite{Bodrato:2007:issac,Nissim:2026:Division}.

We here consider algebras of small dimension and characteristic, using
both computational and theoretical tools to improve the best known
complexity for various cases.

\paragraph{Contributions.}
More precisely, we start with the different representations of
bilinear maps as tensors and matrix subspaces in~\cref{sec:prelim},
where we also improve on a randomized search for additive
complexities.
We then establish the tensor rank of algebras of order~$243$
in~\cref{sec:trank}.
As the multiplication in finite fields or semifields is encoded by
matrix subspaces, we also consider, in~\cref{sec:SSLP}, upper bounds
on the shortest straight-line programs to compute the matrix-vector
product for the generator matrix of small linear codes over~$\F_2$
and~$\F_3$.
Our exploration of the additive complexities continues
in~\cref{sec:folding} by the definition of a folding technique for the
multiplication of polynomials modulo another polynomial.
Finally, we establish the additive complexity of algebras of
order~$81$ in~\cref{sec:deg4} and improve on those of algebras of
order~$32$ and~$243$ in~\cref{sec:deg5}. This allows us also give
improved algorithms for multiplications in extensions of degree~$5$
over~$\F_3$.

\section{Preliminaries}\label{sec:prelim}

Let us regard $\F^n$ as the space of column vectors over the field $\F$, and $(\F^n)^*$ as the space of row vectors.

\begin{definition}
    Consider $T\in (\F^n)^*\otimes (\F^n)^*\otimes F^n$. The {\em tensor rank} of $T$ is the minimum integer $r$ such that there exists a decomposition $T=\sum_{i=1}^rL_i\otimes R_i\otimes P_i$.
\end{definition}

\begin{definition}\label{def:LRPRepresentation}
The \emph{\LRP} representation of a tensor
decomposition~\({T=\sum_{i=1}^{r}{L_i}\tensorproduct{R_i}\tensorproduct{P_i}}\)
is formed by the three matrices~\(\mat{L}\)
in~\(\Matrices{\F}{\matrixsize{r}{n}}\),~\(\mat{R}\)
in~\(\Matrices{\F}{\matrixsize{r}{n}}\),
and~\(\mat{P}\)
in~\(\Matrices{\F}{\matrixsize{n}{r}}\), such
that $L_i$ is the $i$-th row of $L$, $R_i$ is the $i$-th row of $R$, and $P_i$ is the $i$-th column of $P$.
\end{definition}

$L$ is the \emph{Left-matrix},
$R$ is the \emph{Right-matrix},
$P$ is the \emph{Post-matrix} (or the \emph{Product-matrix}) of an
\LRP representation. We will write $T=T(L,R,P)$.

From a tensor (decomposition) we can define a bilinear, but not
necessarily associative, multiplication on $\F^n$.
\[
x\star_Ty := \sum_{i=1}^r (L_i\cdot x)(R_i\cdot y)P_i.
\]
We can write this in terms of the matrices $L,R,P$ as
\[
x\star_Ty = P(Lx \odot Ry),
\]
where $\odot$ is the Hadamard (entrywise) product of vectors.
Hence the overall complexity of computing $x\star_T y$ is determined
by the tensor rank $r$, and the number of operations required to
calculate the three matrix-vector products $Lx$, $Ry$, and $P(Lx\odot
Ry)$. In other words, we seek to optimise the rank of tensor
decompositions and the complexity of the associated {\em straight line
  programs}.

\subsection{Fields and Semifields}

An $\F$-bilinear multiplication $\star$ on $\F^n$ defines a {\it
  presemifield} if $x\star y=0$ if and only if $x=0$ or
$y=0$. Equivalently, every linear equation of the form $a\star x=b$ or
$y\star a=b$ have unique solutions. Note that we do not assume any
associativity or commutativity of $\star$. If $\star$ possesses an
identity element, we call it a {\em semifield}. If it is furthermore
both associative and commutative, it is a field. We refer to
\cite{Lavrauw:2011:finitesem} for background.

Finite semifields have been classified for some small parameters, and many constructions exist; see for example \cite[Table 4]{Combarro:2011:IJCM},  \cite{Rua:2009:sem64}, \cite{Sheekey:2019:MRDsurvey}, \cite{Lavrauw:2011:finitesem}. The study of semifields from the point of view of tensors was initiated in \cite{Liebler:1981:tensor} and furthered in \cite{Lavrauw:2013:tensor}.

In
\cite{Lavrauw:2022:tranksemi} it was shown that there exist semifields
with lower tensor rank (multiplicative complexity) than the field of
the same order. This suggests that semifields may allow better
performance in applications where the associativity of a field is not
necessary. In this work we consider the overall complexity of small
fields and semifields.

\subsection{Isotopy and Isotropy}

We define an action on tensors via invertible linear maps. Although the definition does not depend on a decomposition, it will be more convenient for us to work directly with a decomposition.

\begin{definition}
Let  \({T=\sum_{i=1}^{r}{L_i}\tensorproduct{R_i}\tensorproduct{P_i}}\), and let $(X,Y,Z)$ be a triple of invertible linear maps. Then we define
\[
T^{(X,Y,Z)} := \sum_{i=1}^r(L_iX)\otimes (R_iY)\otimes (ZP_i)
\]
Two tensors $T$ and $S$ are {\em isotopic} if there exists a triple $(X,Y,Z)$ such that $S=T^{(X,Y,Z)}$. A triple $(X,Y,Z)$ is an {\em isotropy} of $T$ if $T^{(X,Y,Z)}=T$.
\end{definition}

Note that $T(L,R,P)^{(X.Y,Z)}= T(LX,RY,ZP)$. A tensor $T$ defines a (pre)semifield if and only if $T^{(X.Y,Z)}$ does too. Isotopy is one of the natural notions of equivalences for semifields and tensors, along with the action of $S_3$ permuting the factors.

The methods of \cite{Karstadt:2017aa} includes the possibility of using an isotopy to exchange $T$ for another tensor with smaller additive complexity; savings are made when a multiplication is used recursively, since the change of bases $(X,Y,Z)$ contribute less. However for the case of a fixed algebra, this will always be more expensive, and so we need to consider the additive complexity of a fixed algebra rather than utilizing an isotopy. We can utilize isotropies to our advantage.

\subsection{\LRP representation and straight-line
  programs}\label{sec:LRP-SLP}

While the multiplicative complexity of an algebra is invariant up to
isotopy, the additive complexity is not. We must consider the \LRP
representations and optimize the number of operations required. That
is, we seek a {\em straight line program} (\SLP) of shortest length.

\begin{example}\label{ex:Karatsuba}
For instance, denoting by $\odot$ the Hadamard product,
the following version of Karatsuba's algorithm for the
multiplication of degree $1$ polynomials, with $3$ base field multiplications,
$c_0+c_1X+c_2X^2=(a_0+a_1X)(b_0+b_1X)=a_0b_0+(a_0b_0+(a_0-a_1)(b_1-b_0)+a_1b_1)X+a_1b_1X^2$,
is represented by~\cref{eq:LRPKaratsuba}.
\begin{equation}\label{eq:LRPKaratsuba}
\begin{bmatrix}c_0\\c_1\\c_2\end{bmatrix}=
\underbracket{\begin{bmatrix}1&0&0\\1&1&1\\0&0&1\end{bmatrix}}_{P}
\left(
\underbracket{\begin{bmatrix}1&0\\1&-1\\0&1\end{bmatrix}}_{L}
\begin{bmatrix}a_0\\a_1\end{bmatrix}
\odot
\underbracket{\begin{bmatrix}1&0\\-1&1\\0&1\end{bmatrix}}_{R}
\begin{bmatrix}b_0\\b_1\end{bmatrix}
\right)
\end{equation}
\end{example}

The \plinopt
software~\cite{jgd:2026:plinopt,jgd:2024:accurate,Dumas:2025ab},
produces straight line programs (\SLP) optimizing its arithmetic cost
(additions and scalar multiplications) by several techniques of common
sub-expression elimination~\cite[Alg.~1]{Dumas:2025ab}, together with
a kernel optimization~\cite[Alg.~2]{Dumas:2025ab}.

We improve here this kernel optimization in~\cref{alg:hybridkernel}.
The idea of~\cite[Alg.~2]{Dumas:2025ab} is to separate the output
computations into two blocks. Starting from the matrix, select a set
of independent rows and a set of rows depending on the first set.
First, produce an optimized \SLP for the first set; this \SLP computes
a subset of the ouput. Then second, instead of optimizing the second
set, optimize their expression as linear combinations of the first
subset of the ouput. If the linear dependencies are much sparser or
structured that the initial rows, then this second \SLP can be faster.
\plinopt then searches for the best random selection of the rows in
the first set.
We propose here two improvements on this technique:
\begin{itemize}
\item Allow extra standard basis vectors in the first set: they do not
  need any operations to be computed, and allow the second set rows to
  be expressed with \emph{combinations} of the input and of the first
  ouput (while the original algorithm expressed them with only the
  input or only the first output); in particular the method now
  applies even if the initial matrix is not full rank, or has more
  columns than rows;
\item Allow the first group to have more rows than the rank of the
  matrix: again this increases the number of possibilities but widens
  the search area.
\end{itemize}
These techniques are formalized in~\cref{alg:hybridkernel} and have
been proven useful for finite fields and semifields (for instance it
was able to find an  \SLP in \cref{sssec:F243} with one operation less
than what~\cite[Alg.~2]{Dumas:2025ab} previously found).
\Cref{alg:hybridkernel} now also allows to apply the Kernel technique
to matrices with more columns than rows.

\begin{algorithm}[ht]
\caption{Hybrid randomized kernel decomposition of a matrix}\label{alg:hybridkernel}
\begin{algorithmic}[1]
\REQUIRE{\(\mat{M}\) in~\(\F^{m{\times}n}\) such that~\({r=\MatrixRank{\mat{M}}}\).}
\ENSURE{A straight line program
  computing~\(\vec{u}\leftarrow{\mat{M}}{\cdot}\vec{v}\) for~$\vec{v}\in\F^n$.}

\STATE Randomly select~$k\in{0..r}$;
\STATE $\mat{G}\in\F^{k{\times}n}\leftarrow$ random selection of~$k$
independent rows of~$\mat{M}$;
\STATE $\mat{H}\in\F^{(n-k){\times}n}\leftarrow$ random selection
of~$n-k$ distinct vectors of the canonical basis of
dimension~$n$, independent of~$\mat{G}$;
\STATE Randomly select~$j\in{0..(m-k)}$;
\STATE $\mat{S}\in\F^{j{\times}n}\leftarrow$ random selection of $j$ rows
of $\mat{M}\setminus{\mat{G}}$;
\STATE $\mat{R}\in\F^{(m-j-k){\times}n}\leftarrow{\mat{M}\setminus(\mat{G}\cup{\mat{S}})}$;
\hfill\COMMENT{$\mat{M}=\begin{smatrix}\mat{G}\\\mat{R}\\\mat{S}\end{smatrix}$,
up to row-perm.}
\STATE $P_1\leftarrow$ \SLP for $\begin{smatrix}\mat{G}\\\mat{R}\end{smatrix}$;
\hfill\COMMENT{Via~\cite[Alg.~1]{Dumas:2025ab}}
\STATE Solve $\mat{S}=\mat{X}\begin{smatrix}\mat{G}\\\mat{H}\end{smatrix}$,
for~$\mat{X}\in\F^{j{\times}n}$;
\STATE $P_2\leftarrow$ \SLP for $\mat{X}$;
\hfill\COMMENT{Via~\cite[Alg.~1]{Dumas:2025ab}}
\RETURN $P_1$ followed by $P_2$ applied on:
\begin{itemize}
\item the first $k$ output of $P_1$; \hfill\COMMENT{computing ${\mat{G}}{\cdot}\vec{v}$}
\item the $n-k$ values of $\vec{v}$ corresponding to $H$.
\end{itemize}
\end{algorithmic}
\end{algorithm}

\subsection{Subspace of matrices from tensors and decompositions}

From an \LRP representation of a $3$-tensor $T$ we can define subspaces of $2$-tensors (or matrices) in some natural ways. We identify $L_i\otimes R_i$ with the rank-one matrix $L_i^{\intercal}R_i$.
\begin{align*}
    W(L,R)&:= \langle L_i\otimes R_i:i\in [1,r]\rangle,\\
    U(L,R,P) &:= \left\langle \sum_{i=1}^r P_{ji} (L_i\otimes R_i):j\in [1,n]\right\rangle.
\end{align*}

Clearly we have that
\[
U(L,R,P)\leq W(L,R)\leq \F^{n\times n},
\]
Note that $U(L,R,P)$ does not depend on the particular \LRP
representation, and so we can define $U(T)$. Not that this is one of
the {\em contraction spaces} of the tensor $T$. A tensor $T$ defines a
field or (pre)semifield multiplication if and only if $U(T)$ is
$n$-dimensional and every nonzero element of $U(T)$ is invertible;
such a space is called a {\it semifield spread set}. Note that this
notion coincides with that of a {\it maximum rank-distance (MRD) code}
of minimum distance~$n$.

On the other hand, the space $W(L,R)$ does depend on the particular
representation, and it is well known that the tensor rank can be
computed by finding a space $W$ spanned by rank one tensors of
smallest dimension such that $U(T)\leq W$. From such a space $W$ we
can construct $L$ and $R$ by choosing a basis $\{B_i\}$ of rank one
matrices for $W$, letting $L_i^{\intercal}$ (resp. $R_i$) be a basis element for
the column (resp. row) space of $B_i$, and $P_j$ the coordinate vector
of the $j$-th element of a basis of $U(T)$ with respect to $\{B_i\}$.

Note also that for any invertible $X,Y,Z$ we have $U(LX,RY,ZP)=
X^{\intercal}U(L,R,P)Y$, and $W(LX,RY)=X^{\intercal}W(L,R)Y$. In this case, we say that $U(LX,RY,ZP)$
and $U(L,R,P)$ are {\it equivalent}. Furthermore, for any $r\times r$
permutation matrix $\sigma$ and any invertible $r\times r$ diagonal
matrices $D_1,D_2$ we have $W(\sigma D_1L,\sigma D_2R)=W(L,R)$, since
the effect of $\sigma$ is to permute the spanning set $L_i\otimes
R_i$, and multiplication by a diagonal matrix multiplies $L_i\otimes
R_i$ by a scalar. Hence if $L'=\sigma D_1LX$ and $R'=\sigma D_2LY$,
then $W(L,R)$ and $W(L',R')$ are equivalent.

\section{Multiplicative complexity}\label{sec:trank}
\subsection{Tensor ranks}

The following theorem of Brockett and Dobkin \cite{brockett}, in combination with data on the best known linear codes from \cite{Grassl:codetables}, gives lower bounds on the tensor rank of any field or semifield. If $T$ is the tensor of multiplication of a field or semifield of order $q^n$, then each of the matrices $L^{\intercal},R^{\intercal},P$ from a rank $r$ decomposition of $T$ must be generator matrices for an $[r,n,n]_q$-linear code. Let $A_q(n,n)$ denote the shortest length of a linear code of dimension $n$ and minimum distance at least $n$ in the Hamming metric over $\F_q$.

\begin{theorem}[\cite{brockett}]\label{thm:coding}
    The tensor rank of a field or semifield of order $q^n$ over $\F_q$ is at least $A_q(n,n)$.
\end{theorem}

When $q\geq 2n-1$, we have $A_q(n,n)=2n-1$, and the tensor rank of the field $\F_{q^n}$ is $2n-1$. When $q$ is small, the question remains open in general.

Here we tabulate the  known results on tensor ranks over $\F_2$ and $\F_3$. For extensions of degree $3$, the question was fully resolved in \cite{Lavrauw:2013:3x3x3}. For polynomial multiplication over $\F_2$, see \cite{Chen:2025:Flip}.

\begin{table}[htb]\centering
\caption{Table of known tensor ranks}
\label{tab:known}
\begin{tabular}{lcc}
\toprule
Set & t. rk & reference \\
\midrule
$\F_{2^4}$ & 9 & \cite{karatsuba:1963:multiplication},\cite{Chudnovsky:1988:complex}\\
$\SF_{2^4}$ \#2, \#3 & 9 & \cite{Lavrauw:2022:tranksemi}\\
\hline
$\F_{3^4}$,$\SF_{3^4}$ \#9  & 9 & \cite{karatsuba:1963:multiplication},\cite{Lavrauw:2013:tensor}\\
$\SF_{3^4}$ \#1,$\ldots$,\#8,\#10,\#11 & 8 & \cite{Lavrauw:2022:tranksemi}\\
\hline
$\F_{2^5}= \SF_{2^5}$ \#5  & 13 & \cite{Montgomery:2005:five}\\
 $\SF_{2^5}$ \#1,$\ldots$,\#4,\#6 & $ 13$ & Section \ref{sec:S32}\\
\hline
$\F_{3^5}$  & 11 & \Cref{thm:f243} \\
$\SF_{3^5}$ \#2, \#4 & 10 & \Cref{thm:veronese243} \\
$\SF_{3^5}$ \#3, \#5,\#6 & 11 & \Cref{thm:veronese243} \\
$\SF_{3^5}$ \#7,$\ldots$,\#12 & $\geq$11 & \Cref{thm:veronese243} \\
\hline
$\F_{2^6}$  & 15 & \cite{Grassl:codetables},\cite{Cenk:2010:ffmulcurve}\\
$\F_{3^6}$  & $\in [12,15]$ & \cite{Grassl:codetables},\cite{Cenk:2010:ffmulcurve}\\
\hline
$\F_{2^7}$  & $\in [18,22]$ & \cite{Grassl:codetables},\cite{Cenk:2010:ffmulcurve}\\
\hline
$\F_{2^8}$  & $\in [20,24]$ & \cite{Grassl:codetables},\cite{Cenk:2010:ffmulcurve}\\
\bottomrule
\end{tabular}
\end{table}

\subsection{Fields and semifields of order $243$}

There are $12$ isotopy classes of finite semifields of order $243$ \cite{Rua243}.
By~\cref{thm:coding}, the tensor rank of any such algebra is at least
$A_3(5,5)=10$. Moreover, there is a unique equivalence class of
$[10,5,5]_3$ codes \cite{vanEupen}; let $C$ denote one such code, and
$G$ a generator matrix for $C$. Hence any rank $10$ decomposition of a
corresponding tensor must have that $L^{\intercal},R^{\intercal},P$ are generator matrices
for codes equivalent to $C$; that is, $L=\sigma DGX$ for some
permutation matrix $\sigma$, some invertible diagonal matrix $D$, and
some invertible matrix $X$.

\begin{theorem}\label{thm:f243} The tensor rank of the finite field with~$3^5$
  elements is~$11$.
\end{theorem}
\begin{proof}
From \cite{Barbulescu:2012:bilmaps}, the tensor rank of $\F_{3^5}$ is at most $11$. So it remains to prove that the tensor rank cannot be $10$.

We construct spaces $W(G,\sigma G)$ for $\sigma\in S_{10}$. However we can reduce our considerations further, since if $\sigma DGX=G$, then $W(G, G)$ and $W(G,\sigma G)$ are equivalent. So we must check only representatives for the cosets of a subgroup of $S_{10}$. Direct calculation with MAGMA gives that we need to check only $2520$ representatives, a reduction by a factor of $720$.

For each such space, we exhaustively search for semifield spread sets using a modification of the methods of \cite{Rua:2009:sem64}, and classify them up to isotopy. We find no occurrences of the field $\F_{3^5}$, proving the claim.
\end{proof}

Next we adapt the algorithm of \cite{Barbulescu:2012:bilmaps} to the
cases of finite semifields. This algorithm searches for \LRP
decompositions where $L=R$. We refer to this as the {\em partially
  symmetric tensor rank}. Necessarily any semifield with such a
decomposition must be commutative, and so we apply it to those
labelled $\#1,\cdots\#7$ of \cite{Rua243}.

\begin{theorem}\label{thm:veronese243}
    The partially symmetric tensor ranks of the commutative semifields with 243 elements are as listed in~\cref{tab:sf243bdez}.
\end{theorem}

\begin{table}[htb]\centering
\caption{Partially symmetric tensor rank of semifields with $243$ elements via \cite{Barbulescu:2012:bilmaps}}
\label{tab:sf243bdez}
\begin{tabular}{lcc}
\toprule
Set & p. symm. t. rk & \# solutions (bases) \\
\midrule
$\F_{3^5}$ & 11 & 121\\
$\SF_{3^5}$ \#2 & 10 & 11\\
$\SF_{3^5}$ \#3 & 11 & 121\\
$\SF_{3^5}$ \#4 & 10 & 1\\
$\SF_{3^5}$ \#5 & 11 & 22\\
$\SF_{3^5}$ \#6 & 11 & 31\\
$\SF_{3^5}$ \#7 & 12 & $1\,310\,980$\\
\bottomrule
\end{tabular}
\end{table}

This does not preclude there existing other semifields of tensor rank $10$; however the following result shows that this does not occur.

\begin{theorem}
There exist precisely two classes of semifields of order $3^4$ of tensor rank~$10$.
\end{theorem}
\begin{proof}
Following the proof of~\cref{thm:f243}, we find that spread sets equivalent to those arising from the semifields labelled $\#2$ and $\#4$ do occur as subspaces of $W(G,\sigma G)$, and no others.
\end{proof}

We can also use this setup to give a further proof to the fact that there is no partially symmetric algorithm with $L=R$ to compute multiplication in $\F_{3^6}$. This was established computationally in \cite{Barbulescu:2012:bilmaps}, where the calculation took $4.50\times 10^7$ seconds; here we require no search at all.

\begin{lemma}
  The partially symmetric tensor rank of the field~$\F_{3^6}$, and any
  semifield of order~$3^6$, is at least~$13$.
\end{lemma}

\begin{proof}
  We require that $L^{\intercal}=R^{\intercal}$ and $P$ are generator
  matrices for $[12,6,6]_3$ codes.
  There is a unique $[12,6,6]_3$ code up to equivalence, namely the
  extended ternary Golay code \cite{Pless:1968:Golay}. Thus we may assume that
  $L^{\intercal}$ is any generator matrix for this code;
  for example,
\[
L^{\intercal}=
\begin{smatrix}
1 & 0 & 0 & 0 & 0 & 0 & 0 & 1 & 1 & 1 & 1 & 1 \\
0 & 1 & 0 & 0 & 0 & 0 & 1 & 0 & 1 & 2 & 2 & 1 \\
0 & 0 & 1 & 0 & 0 & 0 & 1 & 1 & 0 & 1 & 2 & 2 \\
0 & 0 & 0 & 1 & 0 & 0 & 1 & 2 & 1 & 0 & 1 & 2 \\
0 & 0 & 0 & 0 & 1 & 0 & 1 & 2 & 2 & 1 & 0 & 1 \\
0 & 0 & 0 & 0 & 0 & 1 & 1 & 1 & 2 & 2 & 1 & 0 \\
\end{smatrix}
\]
Directly calculating we see that $\sum_{i=1}^{12}L_i^{\intercal}L_i=0$, and so
$U(L,L)$ is a subspace of dimension $11$. This would require that $P$
generates an $[11,6,6]_3$ code, which does not
exist~\cite{Grassl:codetables}.
\end{proof}

\section{Additive complexity of linear codes}
\subsection{Shortest straight-line program}\label{sec:SSLP}

Since we know that the matrices occurring in any \LRP decomposition of
a field or semifield of dimension $n$ over $\F_q$ must define an
$[r,n,\geq n]$ code, it will be relevant to know the minimum
complexity of a straight-line program to compute the matrix-vector
product for such a matrix with this property.

  \begin{definition}
      We denote by $\SSLP_q(r,k,d)$ (shortest straight-line
      program) the minimum number of additions and scalar
      multiplications required to compute $Lv$, where $L^{\intercal}$ generates
      an $[r,k,d]_q$ Hamming metric code.
  \end{definition}

Note that for $\F_2$ and $\F_3$ all operations can be achieved by
addition, regarding negations as trivial.

\begin{lemma}\label{lem:minoverall}
    The minimum number of operations needed to compute multiplication
    in a field or semifield of degree $n$ over $\F_q$ using $r$
    multiplications is at least $rM+(3\SSLP_q(r,n,n)+n-r)A$.
\end{lemma}

\begin{proof}
    This follows from the definition of $\SSLP_q$ regarding
    $\Transpose{L}$, $\Transpose{R}$ and $P$ as generator matrices for codes,
    and from the transposition principle: computing $Pv$ requires
    precisely $n-r$ operations more than computing $P^{\intercal}u$.
\end{proof}

Not that for $r$ large compared to $n$ this can be trivial; for example we have $\SSLP_q(n^2,n,n)=0$, achieved by a repetition code, which occurs as the $L$ and $R$ matrix for the standard algorithm for polynomial multiplication.

For small parameters we can perform a randomized search, via
e.g.~\cref{alg:hybridkernel} and \plinopt, to upper bound the value of
$\SSLP_q(r,n,n)$; exhaustive search is possible only in the very
smallest cases.

\begin{table}[htb]\centering
\caption{Table of some shortest straight-line programs for linear codes}
\label{tab:knownSSLP}
\begin{tabular}{lcc}
\toprule
$[r,k,d]_q$ & $\SSLP$ & reference \\
\midrule
$[8,4,4]_2$ & 6 & \cref{thm:SSLP844}\\
$[9,4,4]_2$ & 5 & \cref{lem:944} \\
$[10,4,4]_2$ & 4 & \cref{lem:10:4:4} \\
$[13,5,5]_2$ & 8 & \cref{lem:SSLP1355} \\
\midrule
$[8,4,4]_3$ & 6 & \cref{thm:SSLP844}\\
$[9,4,4]_3$ & 5 & \cref{lem:944} \\
$[10,4,4]_3$ & 4 & \cref{lem:10:4:4} \\
$[10,5,5]_3$ & $\in [5,12]$ & \cref{no repeated columns} and \cref{tab:multF243}\\
$[11,5,5]_3$ & $\leq{12}$ & \cref{tab:multF243} \\
$[12,5,5]_3$ & $\leq{12}$ & \cref{tab:multF243} \\
$[13,5,5]_3$ & $\leq{11}$ & \cref{tab:multF243} \\
\bottomrule
\end{tabular}
\end{table}

We can search exhaustively or randomly an efficient straight-line program for an $[r,k,d]_q$ code as follows.
\begin{itemize}
    \item Begin with a matrix consisting of the standard basis vectors in $\F_q^k$.
\item Add a row as a sum or difference of two previous rows, or as a repeat of a previous row.
\item Repeat until we obtain a code with minimum distance at least $n$, or until we reach a chosen maximum number $r_{max}\geq r$ of rows.
\item Compute the minimum distance of the code obtained by selecting every multiset of $r$ rows. If this is at least $d$, then $\SSLP_q(r,k,d)\leq r_{max}-k-\#\text{repeated rows}$.
\end{itemize}

Note that it is not necessarily the case that the lowest overall
complexity is achieved by an algorithm of lowest multiplicative
complexity, nor that the existence of an $[r,n,n]_q$ code implies an
algorithm of tensor rank $r$. It is also not necessarily the case that
a linear code with minimal $\SSLP$ can occur in the \LRP representation
of some field or semifield.

\subsection{Calculating $\SSLP_q(r,k,d)$ for small parameters}

In this section we prove some theoretical results on the shortest
straight-line program for certain codes. We start with some general
results.

\begin{lemma}\label{weight n needs n-1}
A vector $v$ of weight $s$ can not be computed with less than $s-1$ additions.
\end{lemma}
\begin{proof}
Assume that while computing $v$, the vectors $v_1,\dots,v_l$ have been
computed using the first $l$ additions.
The supports $S_1,\dots,S_l$ of $v_1,\dots,v_l$ are partially ordered
by inclusion.  For each $l$ define $R_l=\{S_{i_1},\dots,S_{i_t}\}$ as
the set of maximal elements in $S_1,\dots,S_l$.
Assume that $v$ is obtained using $k$ additions, namely $v=v_k$, and
that $k$ is the minimum possible, namely that $v$ can not be obtained
using less than $k$ additions.
In this case, $R_k=\{\textrm{Supp}(v_k)\}$ must hold.
Indeed, if this is not the case, some additions have been used to
compute coordinates that we do not need to compute $v$, and we could
compute $v$ with strictly less than $k$ additions.
Now each addition either increases, decreases, or keeps constant the
size of $R_i$.
Since $\{S_1,\dots,S_l\}\subseteq\{S_1,\dots,S_{l+1}\}$ and only one
further addition occurred, the following are the only possibilities:
\begin{enumerate}
\item $|R_l|=|R_{l+1}|$. In this case, $S_{l+1}$ either contains or is
  contained in one of $\{S_1,\dots,S_l\}$. Using this addition, at
  most one new coordinate is computed. If this was not the case,
  $S_l=\{i,j\}$, where $i,j\notin\cup_{i=1}^l S_i$. Since it cannot be
  that $S_i\subseteq \{i,j\}=S_{l+1}$, $S_{l+1}\in R_{l+1}$, against
  $|R_l|=|R_{l+1}|$.
\item $|R_{l+1}|=|R_l|+1$. In this case, $S_{l+1}$ is not contained in
  any of $\{S_1,\dots,S_l\}$, and has at most two coordinates that are
  not in $\cup_{i=1}^l S_i$.
\item $|R_{l+1}|=|R_l|-1$. In this case, $S_{l+1}$ is the disjoint
  union of two elements $S,S'\in R_l$. Moreover, no new coordinate can
  be computed, namely $\cup_{i=1}^l S_i=\cup_{i=1}^{l+1} S_i$.
\end{enumerate}
Since, $R_k=\{\textrm{Supp}(v_k)\}$, for any addition other than the
first, say the $i$-th, such that $|R_i|=|R_{i-1}|+1$, there must be an
operation, say the $j$-th, such that $|R_j|=|R_{j-1}|-1$.
Let $k_l$ be the number of coordinates computed in some vector until
the addition number $l$, namely $k_l=|\cup_{i=1}^l S_i|$.
Of course, $k_0=0$ and $k_1=2$. Say that after the first addition
there are $a$ further additions of the first type, $b$ of the second,
and $c$ of the third. Then $k_{a+b+c+1}\leq a+2b$. Since $b=c$,
$k_{a+b+c+1}\leq a+b+c$, which proves the claim.
\end{proof}
\begin{corollary}\label{weight n in n-1 behaves well}
If a vector $v$ has weight $s$ and is computed in precisely $s-1$
additions, all the $s-1$ intermediate vectors computed share, up to
the sign, the common non-zero entries.
In other words, the submatrix supported on the intersection of the
supports of any two intermediate vectors has rank one.
\end{corollary}
\begin{proof}
Write, as in the proof of~\cref{weight n needs n-1},
$s-1=1+a+b+c$, where $a,b,c$ are the number of additions of type $1,2$
and $3$, respectively, performed after the first addition.
Then each type $(1)$ addition must effectively compute a new
coordinate, and each type $(2)$ addition must add two new coordinates,
as otherwise, in each case, we would get $\textrm{weight}(v)\leq
s-1$.
This means that a type $(1)$ addition, say at step $l+1$, is of the
form $v_{i}+E_j$, where $i\leq l$, $E_i$ is a vector of the canonical
basis, and the $i$-th coordinate does not belong to
$\cup_{i=1}^l\mathrm{Supp}(v_i)$.
Similarly, a type $(2)$ addition, say at step $l+1$, is of the form
$E_i+E_j$, where $\{i,j\}\cap
(\cup_{i=1}^l\mathrm{Supp}(v_i))=\emptyset$.
\end{proof}
The following results are on $[2k,k,k]_q$ codes, with $q\in\{2,3\}$.
Note that, by the Griesmer bound, if a $[2k,k,k]_q$ code with $q\in\{2,3\}$ exists, then $q=2$ and $k\leq 4$ or $q=3$ and $k\leq 6$.
\begin{lemma}\label{no repeated cols for systematic}
Let $G=(I_k|M)$ be the generator matrix of an $[2k,k,k]_q$ code, where
$I_k$ is the $k\times k$ identity matrix and $q\in\{2,3\}$.
Then, if $M$ has a repeated column, then the columns of $G$ can not be
computed in $2k-3$ additions.
Moreover, no vector of the canonical basis is repeated in $G$.
\end{lemma}
\begin{proof}
The second claim is straightforward to check, so we only prove the
first claim.
Suppose by contradiction that we can compute $G$ using $2k-3$
additions and, without loss of generality, that the first two columns
are equal, $M^1=M^2$.
Since no codeword of $G$ has weight less than $k$, no row of $M$ has
two $0$'s: and there are no indexes $i,j$ such that three columns
of~$M$ have the same coordinates in the rows $i$ and $j$ (otherwise
their difference has weight less than $k$).
Moreover, no column of $M$ has two $0$'s. In fact, if rows $i,j$ of a
column of $M$ are $0$, each of the remaining $k-1$ columns of $M$ will
have non-zero entries in the rows $i,j$, and a linear combination of
rows $i,j$ has weight at most $k-1$.
In particular, since $M^1$ is repeated, it has weight $k$.
Since no row and no column of $M$ can have two zero entries, there are
at most $k-2$ zero entries in the rows of $M$, and therefore at least
two rows of $M$ have weight $k$. A linear combination of the
corresponding rows of $G$ gives a codeword of weight less than $k$, a
contradiction.
\end{proof}
\begin{lemma}\label{systematic 2k,k,k codes}
Let $G=(I_k|M)$ be the generator matrix of an $[2k,k,k]_q$ code, where $I_k$ is the $k\times k$ identity matrix and $q\in\{2,3\}$.
Then the columns of $G$ can not be computed in $2k-3$ additions. Equivalently, computing $G$ requires at least $2k-2$ operations.
\end{lemma}
\begin{proof}
Suppose, by contradiction, that we can compute $G$ in $2k-3$ additions.
Arguing as in the proof of \Cref{no repeated cols for systematic}, we
can assume that no row of $M$ has two $0$ entries, that no pair of
distinct indexes $i,j$ such that three distinct columns $M$ have the
same coordinates in the rows $i,j$ exists, and that no column of $M$
has two zero entries.
As a consequence, each column has weight at least $k-1$, and
by~\cref{weight n needs n-1} can be computed in no less than $k-2$
additions.
Assume that each column of $M$ has weight $k$.
Then, by~\cref{weight n  needs n-1}, using the first $k-1$ additions,
at most one column of $M$ is computed, and at least one more addition
is needed to compute each of the remaining $k-1$ columns of
$M$ (all these columns have full weight, but must be all distinct --
by the sign of their elements).
Therefore, at least $k-1+k-1=2k-2$ additions are needed, a
contradiction.
We can therefore assume that at least one column of $M$ has a zero
entry, and that at addition number $k-2$, one column of $M$ has been
computed.
Without loss of generality, write
$P^1,\dots,P^{k-3},P^{k-2}=M^1=\Transpose{\begin{pmatrix}\star&\dots&\star&0\end{pmatrix}}$,
where $\star$ indicates a non-zero scalar, for the $k-2$ vectors
computed with the first $k-2$ additions.
By~\cref{weight n in n-1 behaves well}, all the vectors
$P^1,\dots,P^{k-2}$ share, up to a non-zero scalar multiple, all the non-zero
coordinates.
Since in the first $k-2$ additions the column $M^1$ was the only one
computed, each of the following $k-1$ additions must compute a new
column of $M$.
Since no row of $M$ has two zero entries, the next addition must be
adding a non-zero entry as last coordinate. In particular, $M^1$ and
$M^2$ share, up to the sign, all but one entries.
Each $M^i$ is obtained either combining two vectors among
$P^1,\dots,P^{k-3},P^{k-2}=M^1, M^2$, or combining one of those vector
with a vector of the canonical basis.
In either case, they must share with $M^1$ at least two coordinates.
Now, $M^3$, by the properties above, shares at least two non-zero
entries with $M^1$, and therefore with $M^2$, contradicting the fact
that no distinct indexes $i,j$ exist such that three columns of $M$
have same entries in row $i,j$.
\end{proof}
\begin{lemma}\label{no repeated columns}
The generator matrix $G$ of a $[2k,k,k]_q$ code, where $q\in\{2,3\}$ and $k\geq 4$, has no repeated columns.
\end{lemma}
\begin{proof}
Let $C$ be the $[2k,k,k]_q$ code generated by $G$. Assume to the
contrary that $G$ has two identical columns. Take the {\it shortening}
of this code with respect to these two columns; that is, the subcode
of co-dimension $1$ containing all elements of $C$ with zeroes in
these two positions. Deleting these two positions gives an
$[2k-2,k-1,\geq k]_q$ code. In particular, this is an MDS code.
For $q\in\{2,3\}$, MDS codes of dimension at least $3$ have minimum distance $d\leq 2$, so no such code exists.
\end{proof}

\begin{theorem}\label{th:at least 6 adds}
The generator matrix $G$ of a $[8,4,4]_q$ code, where $q\in\{2,3\}$, requires at least $6$ additions.
\end{theorem}
\begin{proof}
If $G$ is systematic, Lemma \ref{systematic 2k,k,k codes} proves the claim.
Assume by the contrary that $G$ can be computed in $5$ additions.
The only possibility left, by \Cref{no repeated columns}, is that $3$
columns of $G$ are vectors of the canonical basis, and each of the $5$
additions computes one column of $G$. Without loss of generality, let
$G=(I_{4\mid 3}\mid M)$, where $I_{4\mid 3}$ is the $4\times 3$ matrix
whose columns are the first three vectors of the canonical basis.

Then, at least one column of $M$, say $M^1$, has at least two zero
entries (otherwise the first of the $5$ additions cannot realize of
full column, and there remains $4$ additions only to realize $5$
distinct outputs).
First, if the entry in the last row of $M^1$ is zero, without loss of
generality,
$M^1=\Transpose{\begin{pmatrix}1&\square&0&0\end{pmatrix}}$,
where $\square$ denotes an unknown entry, possibly zero.
Since $G$ generates a code of minimum distance $4$, no row of $M$ has $3$ zero entries.
In particular, for $i=2,3,4,5$, we can assume
$M^i=\Transpose{\begin{pmatrix}\square&\square&\square&1\end{pmatrix}}$.
If
$M^2=\Transpose{\begin{pmatrix}\star&\star&0&1\end{pmatrix}}$, where
$\star$ denotes an unknown non-zero entry, then the third coordinate
of each of $M^i$, $i=3,4,5$, must be non-zero. In this case, a linear
combination of the third and fourth row will have weight less
than~$4$.
The same reasoning applies if any of $M^i$, $i=3,4,5$ has the third
coordinate zero. Therefore, at least two among $M^i$, $i=3,4,5$ share
the same third coordinate, and a linear combination of the third and
fourth row will have weight less than~$4$, which is again not
possible.

Second, w.l.o.g., the last case to consider is
$M^1=\Transpose{\begin{pmatrix}\star&0&0&1\end{pmatrix}}$.
Among the remaining columns of $M$, at most one has a $0$ in the
fourth coordinate.
Assume one has, and w.l.o.g., let it be $M^2$. Assume first
$M^2=\Transpose{\begin{pmatrix}0&y&1&0\end{pmatrix}}$.
Notice that $M^i$, $i=3,4,5$, is such that $M^i=M^1+something$, since
they have a $1$ in the fourth coordinate.
If $M^i=M^1\pm E^1$, then with a combination of the second and third
row we can get a codeword of weight less than~$4$.
If $M^i=M^1\pm E^j$, $j=2,3$ then the first row plus or minus the fourth row
has weight less than~$4$.
Assume now
$M^2=\Transpose{\begin{pmatrix}y&0&1&0\end{pmatrix}}$.
Remember that $M^i$, $i=3,4,5$, is such that $M^i=M^1+something$,
since they have 1 in the fourth coordinate.
Since in the second row there is only 1 non zero entry so far, they
all need to be of the form $M^i=M^1\pm E^2$. In this case, a
combination of the second and fourth row, will have weight less
than~$4$.
Thus $G$ can not be computed in $5$ additions.
\end{proof}

\begin{theorem}\label{thm:SSLP844}
The shortest straight-line program for the generator matrix of an $[8,4,4]_q$ code, where $q\in\{2,3\}$, is $6$; that is, $\SSLP_q(8,4,4)=6$, $q\in\{2,3\}$.
\end{theorem}

\begin{proof}
By \Cref{th:at least 6 adds}, $\SSLP_q(8,4,4)\geq 6$, where $q\in\{2,3\}$.
On the other hand, the following matrix in~\cref{eq:844_6} can be
realized with $6$ additions:\\
\begin{minipage}{.6\columnwidth}
\begin{lstlisting}[style=slp]
o0:=i0; o1:=i1; o2:=i2; o3:=i3; t4:=i0+i2; t5:=i1+i3; o4:=i1+t4; o5:=i0+t5; o6:=i3+t4; o7:=i2+t5;
\end{lstlisting}
\end{minipage}
\begin{minipage}{.35\columnwidth}\vspace{-5pt}
\begin{equation}\label{eq:844_6}
\Transpose{\begin{smatrix}
 1 & 0 & 0 & 0 & 1 & 1 & 1 & 0 \\
 0 & 1 & 0 & 0 & 1 & 1 & 0 & 1 \\
 0 & 0 & 1 & 0 & 1 & 0 & 1 & 1 \\
 0 & 0 & 0 & 1 & 0 & 1 & 1 & 1 \\
\end{smatrix}}
\end{equation}
\end{minipage}
\end{proof}

The $\SSLP_2$ of the two other linear codes over $\F_2$ of dimension
and distance $4$, $[9,4,4]_2$ and $[10,4,4]_2$, are given
in~\cref{sec:sslp2}.

\section{Folding for extensions}\label{sec:folding}
We now turn to small finite field extensions.
Instead of implementing finite field multiplication as polynomial
multiplication followed by polynomial remaindering,
it is better to consider directly the $L,R,P$ representation of the
complete algorithm.
This is obtained by a \emph{folding} of the $P$ matrix of multiplication:
considering the \LRP representation, reducing modulo an
irreducible polynomial of degree $d$, is adding each row, after the
$d+1$-th one, to the previous rows (times the corresponding
coefficients modulo the irreducible polynomial). The details are
in~\cref{alg:folding}.
\begin{algorithm}[ht]
\caption{Folding extension field multiplication}\label{alg:folding}
\begin{algorithmic}[1]
\REQUIRE $P$, post-matrix of polynomial mult. of degree~$d$;
\REQUIRE $I(X)$ irreducible polynomial of degree~$d+1$.
\ENSURE $P^{(I)}$, post-matrix of polynomial mult. modulo~$I(X)$.
\STATE Let $P^{(I)}_{0..d,*}=P_{0..d,*}$;
\hfill\COMMENT{First $d+1$ rows}
\FOR{$i=d+1$ \To $2d$}
\STATE\label{line:irr} $M={X^i\mod{I(X)}}=\sum_{j=0}^d m_j X^j$;
\FOR{$j=0$ \To $d$}
\STATE $P^{(I)}_{j,*} \pe m_j P_{i,*}$;
\hfill\COMMENT{Coeff. of $X^i$ added to $j$-th row, ${\cdot}m_j$}
\ENDFOR
\ENDFOR
\end{algorithmic}
\end{algorithm}
Obviously, another way of viewing~\cref{alg:folding}, is that it is a
matrix multiplication between
the (Sylvester) $(d+1)\times(2d+1)$ matrix of the reduction modulo
$I(X)$ (the identity matrix augmented by the columns of coefficients
in Line~\ref{line:irr})
and the $(2d+1)\times{r}$ $P$-matrix of the multiplication.
From this, an \SLP is directly obtained for the finite field
multiplication. The question is now to find the irreducible polynomial
that generates the \SLP with the least number of operations.

\section{Additive complexity; extensions of degree $4$}\label{sec:deg4}

\subsection{Finite field of order 81}\label{ssec:F81}
For $\F_{3^4}$ the tensor rank is $9$, with a recursive use of two
levels of Karatsuba's algorithm.
From~\cref{ex:Karatsuba}, the \LRP is
given by: $L^{\otimes^2}\in\F_3^{9{\times}4}$,
$R^{\otimes^2}\in\F_3^{9{\times}4}$ and the folding in $\F_3^{4{\times}9}$ of
$P^{\otimes^2}\in\F_3^{7{\times}9}$ via~\cref{alg:folding}.
We tested all $18$ irreducibles of degree $4$ over $\F_3$ and found
that the best one is $1+X+X^2+X^3+X^4$, giving an \SLP with only $11$
operations for the post-matrix, see
\plinoptdata{3o3o3_F81-Karatsuba-9-21_{L,R,P}.s{ms,lp}}.
The overall additive complexity upper bound is then
$5+5+(6+(9-4))=21$, as detailed in~\cref{tab:mult81}. This means that
we can perform multiplication in $\F_{3^4}$ using a total of $30$ operations.

\begin{remark}\label{rk:F81in20}
  It is possible that there could exist an \LRP decomposition for
  $\F_{3^4}$ requiring only $5+5+10=20$ additions, and thus $29$
  operations.
  We calculated all possible matrices using $5$ additions
  giving a $[9,4,4]_3$ code. There are $4104$ such matrices up to
  permutation (two up to code equivalence). An exhaustive search
  through all $L/R$ pairs is too expensive at present.
\end{remark}

\subsection{Semifields of order 81}\label{ssec:S81}
From \cite{Lavrauw:2022:tranksemi}, we know that there exist
semifields of order $81$ with tensor rank $8$. %
By Lemma \ref{lem:minoverall} and Theorem \ref{thm:SSLP844}, we have that that $30=8+6+6+(6+8-4)$ operations
is a lower bound for an algorithm using $8$ multiplications. It turns out that this lower bound is tight, due to the following example. %

\begin{theorem}
    There exists a presemifield of order $3^4$ which achieves optimal multiplicative complexity $8$, and optimal additive complexity for an algorithm with $8$ multiplications.
\end{theorem}

\begin{proof}
See \plinoptdata{3o3o3_S81_8_22_{L,R,P}.s{ms,lp}} for the \LRP matrices for a presemifield. One can verify using the matrices in Appendix \ref{app:s81} that this is isotopic to $\SF_{3^4}$\#1, following the labelling of \cite{Dempwolff:2008:s81}.

Up to sign, these \LRP matrices are similar to that
of~\cref{eq:844_6}. Hence $Lx$ and $Ry$ can be realized with $6$ additions each, and, by the transposition principle, $Pz$ can be
realized with an optimal number of operations of $6+(8-4)=10$. This achieves the bound of \cref{lem:minoverall}.
\end{proof}

Overall, we obtain~\cref{tab:mult81} for the known complexity bounds
of extensions with $81$ elements. Note again that there could exist an algorithm for a presemifield using $9$ multiplications and $20$ additions, but we have not yet found such an example.
\begin{table}[htbp]\centering
\caption{Multiplication in extensions with 81 elements}\label{tab:mult81}
\begin{tabular}{crrr}
\toprule
& rank & ADD & Total\\
\midrule
$\F_{3^4}$ standard  mod $X^4+X-1$ & 16 & 0+0+15 & 16M+15A\\
$\F_{3^4}$ Karatsuba$^{\otimes^2}$ & \multirow{2}{*}{9} & \multirow{2}{*}{5+5+11} & \multirow{2}{*}{9M+21A}\\
(mod $1+X+X^2+X^3+X^4$) & & & \\
\midrule
$\SF_{3^4}$ here  & 8 & 6+6+10 & 8M+22A \\
\bottomrule
\end{tabular}
\end{table}
\section{Additive complexity; extensions of degree $5$}\label{sec:deg5}

\subsection{Multiplying polynomials of degree 4}
As shown in~\cref{sec:folding}, to implement the multiplications in
extensions of degree $5$, we consider their \LRP representation,
obtained by a folding of the $P$ matrix. We thus start by looking at
the multiplication of polynomials of degree $4$.

The standard algorithm for this has rank $25$, while the Toom-Cook-$5$
interpolation method has tensor rank $2d-1=9$, but requires at least
$9$ distinct points in the field.
For very small fields, variants of Karatsuba's method can be found
and the paper~\cite{Montgomery:2005:five} proposes an algorithm of
rank $13$.
The study in~\cite{EGI:2011:africacrypt} then compares the total
number of operations for \SLP{s} of these methods.

We constructed the \LRP matrices and optimized them with \plinopt.
We were able to find programs with slightly less operations
than the ones presented
in~\cite{EGI:2011:africacrypt}.
The algorithms and \SLP for Montgomery's method are in~\cref{app:Deg4}
and there:
\plinoptdata{4o4o8\_Montgomery-13-58_{L,R,P}.s{ms,lp}}.
Now the algorithms for the Toom-$5$ method are the optimal ones,
presented in~\cite{Bodrato:2007:issac}, and can be found there:
\plinoptdata{4o4o8_Toom5_{L,R,P}.s{ms,lp}}.
\Cref{tab:deg4poly} gives the obtained number of operations
(rank is the number of finite field elements multiplications, ADD is
the number of additions, SCA is the number of multiplications by a
constant, sums of $3$ elements are the details for each one of the
$L$, $R$ or $P$ matrices).
For the L and R matrices of both methods, \plinopt finds respectively
$11ADD$ and $19ADD+11SCA$, as expected from
\cite{Montgomery:2005:five,Bodrato:2007:issac}.
Then, on the one hand, for P, \cite{Montgomery:2005:five}
states that "Other repeated subexpressions arise while processing the
outputs of the 13 products, although it is hard to count
these". \plinopt obtains an \SLP with $31$ additions and $5$ scalings,
saving $4$ operations over~\cite{EGI:2011:africacrypt}.
On the other hand, for Toom-5 method, the partial search of
\plinopt, only finds $35ADD+28SCA$ operations, while the optimal \SLP
of~\cite[A.3]{Bodrato:2007:issac}, found by exhaustive search, needs
$32ADD+21SCA$ operations.

\begin{table}[htbp]\centering
\caption{Degree 4 polynomial multiplication}\label{tab:deg4poly}
\begin{tabular}{rrrrr|r}
\toprule
& rank & ADD & SCA & Total & EGI~\cite{EGI:2011:africacrypt}\\
\midrule
Standard & 25 & 0+0+16 & - & 25M+16A & 25M+16A\\
Montg. & 13 & 11+11+31 & 0+0+5 & 13M+58A & 13M+62A\\
Toom-5 & 9 & 19+19+32& 11+11+21 &  9M+113A & 9M+140A\\
\bottomrule
\end{tabular}
\end{table}

Note that in~\cite{EGI:2011:africacrypt}, the count of operations is
done slightly differently, namely with different weights given to
certain combinations of operations. The numbers presented in this
table are obtained by counting each
operation\footnote{\cite{EGI:2011:africacrypt} is counting scalings
  and shifts differently: their number of operations is
  $11ADD+8SHIFTS+4SCAL$ for $L$ and $R$, then
  $28ADD+11SHIFTS+11SCAL+28A$ for P, but they assume that in practice
  $SHIFT\approx\frac{2}{3}ADD$, then $SHIFT\&ADD\approx\frac{5}{4}ADD$
    and that $SCAL\approx{2ADD}$ so that their count is
    $\approx{9M+137A}$ instead.}.

\subsection{Fields and semifields of order 32}\label{sec:S32}
For $\F_{2^5}$, we looked at the $6$ irreducible polynomials of degree
$5$ over $\F_2$. We found that $1+X^2+X^5$ is more interesting for the
Standard algorithm, while $1+X+X^2+X^4+X^5$ is better for Montgomery's
algorithm.
The algorithms can be found in
\plinoptdata{4o4o4_F32_Standard_{L,R,P}.s{ms,lp}}
and
\plinoptdata{4o4o4_F32_Montgomery_{L,R,P}.s{ms,lp}}.
Overall both algorithms have the same number of operations, as shown
in~\cref{tab:deg4F32}, but Montgomery's version has less base field
multiplications.

The remaining semifields of order 32 all have tensor rank $13$.
There are too many \LRP representations to check exhaustively; a
limited search did not return any examples competitive with
$\F_{2^5}$, only one with $3$ more operations, see
\plinoptdata{4o4o4_S32_13_\#2_{L,R,P}.s{ms,lp}},
as shown in the last row of~\cref{tab:deg4F32}.

\begin{lemma}\label{lem:SSLP1355}
  It holds that $\SSLP_2(13,5,5)= 8$. Therefore the total number of
  operations required to perform multiplication in a field or
  semifield of order $32$ is at least $13M+32A$.
\end{lemma}

\begin{proof}
It can be readily verified that the below matrix generates a
$[13,5,5]_2$ code. A straight-line program using $8$ additions for its
transpose can be easily found by hand:\\
\begin{minipage}{.35\columnwidth}\vspace{-5pt}
\[\begin{smatrix}
1 & 0 & 0 & 0 & 0 & 0 & 0 & 1 & 0 & 0 & 1 & 1 & 1 \\
0 & 1 & 0 & 0 & 0 & 0 & 0 & 0 & 1 & 1 & 0 & 1 & 1 \\
0 & 0 & 1 & 0 & 0 & 0 & 1 & 0 & 0 & 1 & 1 & 1 & 0 \\
0 & 0 & 0 & 1 & 0 & 1 & 1 & 0 & 0 & 1 & 0 & 0 & 1 \\
0 & 0 & 0 & 0 & 1 & 1 & 0 & 1 & 1 & 0 & 1 & 0 & 0
\end{smatrix}\]
\end{minipage}\hfill
\begin{minipage}{.6\columnwidth}
\begin{lstlisting}[style=slp]
o0:=i0; o1:=i1; o2:=i2; o3:=i3; o4:=i4; o5:=i4+i3; o6:=i2+i3; o7:=i4+i0; o8:=i4+i1; o9:=o6+i1; o10:=o7+o2; o11:=o8+o10; o12:=o11+o6.
\end{lstlisting}
\end{minipage}

Fewer than $8$ additions would require at least one repeated column. An exhaustive search shows that this is not possible. %
\end{proof}

\begin{table}[htbp]\centering
\caption{Degree 4 folding in characteristic 2}\label{tab:deg4F32}
\begin{tabular}{rrrr}
\toprule
& rank & ADD & Total \\
\midrule
Standard mod 2 & 25 & 0+0+16 & 25M+16A \\
Montgomery mod 2 & 13 & 9+9+19 & 13M+37A\\
\midrule
Standard $\F_{2^5}$ & \multirow{2}{*}{25} & \multirow{2}{*}{0+0+24} & \multirow{2}{*}{25M+24A} \\
(mod $X^5+X^2+1$) & & & \\
$\F_{2^5}$ via Montgomery mod 2 & \multirow{2}{*}{13} & \multirow{2}{*}{9+9+18} & \multirow{2}{*}{13M+36A}\\
(mod $X^5+X^4+X^2+X+1$) & & & \\
\midrule
$\SF_{2^5}~\#2$ & 13 & 12+9+18 & 13M+39A \\
\bottomrule
\end{tabular}
\end{table}

\subsection{Finite field of order 243}\label{sssec:F243}
We further explored the algorithms of rank $11$ for $\F_{3^5}$ found
via the algorithm of~\cite{Barbulescu:2012:bilmaps}.
We noted that the isotropy group of the tensor has order $5\cdot
242^2$, with a group of order $5\cdot 242$ preserving the symmetric
algorithms. This group acts transitively on the $121$ algorithms found
from~\cite{Barbulescu:2012:bilmaps}. This means that
 all of the 121 {\it symmetric} formulae for an instance of
$\F_{243}$ returned have the same additive complexity for $L$ and $R$, and so we only need to optimize one example per irreducible polynomial. The least total number of operations we found was
$11+12+12+(14+11-5)=11M+44A=55$ and the associated algorithms and \SLP,
to multiply elements in $\F_{3^5}$ in 11 multiplications are available
in~\cref{app:F243} and there:
\plinoptdata{4o4o4_F243-11-44_{L,R,P}.s{ms,lp}};
while the ones for Montgomery method are there;
\plinoptdata{4o4o4_F243-Montgomery-13-42_{L,R,P}.s{ms,lp}}.
Interestingly, with its best irreducible ($X^5+X^4-X^3-X^2-1$),
Montgomery's method uses more multiplications, namely $13$ instead of
$11$, but one less operation for each of $L,R$, resulting in the same
number of operations overall.
Thus, as soon as the (base field) multiplication is more expensive
than the (base field) addition, our new algorithm using only $11$
multiplications is preferable.
The best \SLP we could find for an intermediate algorithm of rank
$12$, uses $14+14+(12+(12-5))=47$ additions and is therefore never
interesting.
These results are reported in~\cref{tab:multF243}.

\subsection{Semifields of order 243}
We also explored the algorithms for semifields of order $243$, again
using the algorithm of~\cite{Barbulescu:2012:bilmaps}. We focused on
the semifields of tensor rank $10$:
\begin{itemize}
\item The best algorithms we could find for $\SF_{243}$ \#4 (of
  partially symmetric rank $10$ in~\cref{tab:sf243bdez}), gives:
  $10+12+12+(15+10-5)=10M+44A$ operations;
\item While, trying all the 121 possible formulae for the 11 basis of
  $\SF_{243}$ \#2, the best one gives:
  $10+13+13+(12+10-5)=10M+43A$ operations.
\end{itemize}
This is reported in~\cref{tab:multF243},
and the latter algorithm can be found in
\plinoptdata{4o4o4_S243_10-43__{L,R,P}.s{ms,lp}} and~\cref{app:S243}.

\begin{table}[htbp]\centering
\caption{Multiplication in extensions with 243 elements}\label{tab:multF243}
\begin{tabular}{crrr}
\toprule
& rank & ADD & Total\\
\midrule
$\F_{3^5}$ standard  mod $X^5-X+1$ & 25 & 0+0+24 & 25M+24A\\
$\F_{3^5}$ Montgomery & \multirow{2}{*}{13} & \multirow{2}{*}{11+11+20} & \multirow{2}{*}{13M+42A}\\
(mod $X^5+X^4-X^3-X^2-1$) & & & \\
$\F_{3^5}$ here & \multirow{2}{*}{12} & \multirow{2}{*}{14+14+19}& \multirow{2}{*}{12M+47A} \\
(mod $X^5-X^4+X^3+X^2+X+1$)& & & \\
$\F_{3^5}$ here (mod $X^5-X+1$) & 11 & 12+12+20& 11M+44A \\
\midrule
$\SF_{3^5}$ here  & 10 & 13+13+17& 10M+43A \\
\bottomrule
\end{tabular}
\end{table}

\section{Conclusion and open problems}

In this article we have furthered the study of the tensor rank of fields and semifields, and initiated the study of the additive complexity for semifields. We improve on the state-of-the art in many cases, and compare the new best-known complexities across different algebras. As part of this, we investigate the shortest straight-line programs for linear codes with certain parameters. A number of natural open questions remain.

\begin{itemize}
\item Provide general lower bounds for $\SSLP_q(r,k,q)$, and establish precise results for more small parameters.
    \item Determine whether or not there exists an algorithm for a field or semifield of order $3^4$ requiring only $9M+20A$ operations.
 \item Determine whether or not there exists an algorithm for a semifield of order $3^5$ requiring $10M$ and fewer than $43A$ operations, or $11M$ and fewer than $44A$.
\end{itemize}

\clearpage
\bibliographystyle{plainurl}
\bibliography{bibliography}

\appendix
\section{$\SSLP$ of dimension and distance $4$}\label{sec:sslp2}
We here give the $\SSLP$ for other linear codes
of dimension and distance $4$.

\begin{lemma}\label{lem:944}
$\SSLP_q(9,4,4)=5$, for $q\in\{2,3\}$.
\end{lemma}
\begin{proof}
We consider the generator matrix $G$ of a $[9,4,4]$ code.
If it has a zero column then it requires at least $6$ operations
by~\cref{thm:SSLP844}.
To require fewer than $5$ operations, it thus must have at least $5$
columns that require no operations. W.l.o.g. we can suppose that four
of them are the standard basis vector the rest are repeated columns.
First, suppose there are $2$ repeated columns. Then w.l.o.g. and up to signs,
$G=[I_4\mid{E^1}\mid{E^1}\mid{M}]$ or
$G=[I_4\mid{I_{4\mid{2}}}\mid{M}]$, where $I_{4\mid{2}}$ is the $4\times 2$ matrix
whose columns are the first two vectors of the canonical basis.
To be of weight $4$ the last two rows of $M$ must then be all plus or
minus ones, and at least two among the $3$ are shared between these
two rows.
Then the difference of the last two rows of $G$ is of weight not
larger than $3$.
Now, $G=[I_4\mid{E^1}\mid{\vec{v}}\mid{M}]$.
Second, if the $4$ distinct columns, $\vec{v}$ and those of $M$, all
have $3$ non-zeroes, then they can not be computed in less than $5$
operations (any column requires $2$, and then the $3$ others at least
$1$ operation).
Third, if $\vec{v}=E^1+E^2$, then again the last two rows of
$M$ must be all ones and the difference of the last two rows of $G$ is
of weight no more than $3$.
Now w.lo.g., $G=[I_4\mid{E^1}\mid{E^2+E^3}\mid{M}]$ and the last row
of $M$ is all non-zeroes, for that of $G$ to be of weight $4$.
The other rows of $M$ must have at least $2$ other non-zeroes and must
be distinct in, otherwise the weight of their difference is too small.
Therefore, w.l.o.g. $G$ is as follows:
\[\begin{smatrix}
1&0&0&0&\star&0&\square&\star&\star\\
0&1&0&0&0&\star&\star&\square&\star\\
0&0&1&0&0&\star&\star&\star&\square\\
0&0&0&1&0&0&\star&\star&\star\\
\end{smatrix}\]
Now, the case with the least number of operations is when the stars are
ones and the squares are zeroes. In this case, it does represent a
$[9,4,4]$ code over $\F_2$ or $\F_3$ and its transpose requires
already $5$ operations:\\
\begin{minipage}{.35\columnwidth}\vspace{-5pt}
\[\begin{smatrix}
1&0&0&0&1&0&0&1&1\\
0&1&0&0&0&1&1&0&1\\
0&0&1&0&0&1&1&1&0\\
0&0&0&1&0&0&1&1&1\\
\end{smatrix}\]
\end{minipage}\hfill
\begin{minipage}{.6\columnwidth}
\begin{lstlisting}[style=slp]
o0:=i0; o1:=i1; o2:=i2; o3:=i3; o4:=i0; t4:=i0+i3; o5:=i1+i2; o6:=i3+o5; o7:=i2+t4; o8:=i1+t4;
\end{lstlisting}
\end{minipage}
\end{proof}

\begin{lemma}\label{lem:10:4:4}
$\SSLP(10,4,4)=4$, for $q\in\{2,3\}$.
\end{lemma}
\begin{proof}
If the generator matrix of a $ [10,4,4]$ code has a zero column then
it requires at least $5$ operations by~\cref{lem:944}.
To require fewer than $4$ operations, it thus must have at least
four standard basis and two repeated columns.
On the one hand, suppose there are $3$ repeated columns.
Then w.l.o.g., let $G=[I_4\mid{E^i}\mid{E^j}\mid{E^k}\mid{M}]$, up to signs.
If two among $i,j,k$ are equal, then at least two rows of $M$ are all
non-zeroes and the difference of the same rows of $G$ is of weight not
larger than $3$.
So w.l.o.g, $G=[I_4\mid{I_{4\mid{2}}}\mid{M}]$ and the last row of $M$ is
all non-zeroes.
Then the first $3$ rows of $M$ are distinct of weight at least $2$
and w.l.o.g. $G$ is as follows:\\
\[\begin{smatrix}
1&0&0&0&\star&0&0&\square&\star&\star\\
0&1&0&0&0&\star&0&\star&\square&\star\\
0&0&1&0&0&0&\star&\star&\star&\square\\
0&0&0&1&0&0&0&\star&\star&\star\\
\end{smatrix}\]
Again, the case with the least number of operations is when the stars are
ones and the squares are zeroes. In this case, its transpose requires
already no less than $5$ operations:\\
\begin{minipage}{.35\columnwidth}\vspace{-5pt}
\[\begin{smatrix}
1&0&0&0&1&0&0&0&1&1\\
0&1&0&0&0&1&0&1&0&1\\
0&0&1&0&0&0&1&1&1&0\\
0&0&0&1&0&0&0&1&1&1\\
\end{smatrix}\]
\end{minipage}\hfill
\begin{minipage}{.6\columnwidth}
\begin{lstlisting}[style=slp]
o0:=i0; o1:=i1; o2:=i2; o3:=i3; o4:=i0; o5:=i1; o6:=i2; t4:=i3+i0; o7:=i1+i2+i3; o8:=t4+i2; o9:=t4+i1;
\end{lstlisting}
\end{minipage}
On the other hand, with only two repeated rows, at least $4$
operations are required.
Now, the following matrix also represents a $[10,4,4]$ code over
$\F_2$ or $\F_3$. and its transpose requires only $4$ operations:\\
\begin{minipage}{.35\columnwidth}\vspace{-5pt}
\[\begin{smatrix}
1&0&0&0&0&0&0&1&1&1\\
0&1&0&0&0&0&1&1&1&0\\
0&0&1&0&0&1&1&1&0&0\\
0&0&0&1&1&1&1&0&0&0\\
\end{smatrix}\]
\end{minipage}\hfill
\begin{minipage}{.6\columnwidth}
\begin{lstlisting}[style=slp]
o0:=i0; o1:=i1; o2:=i2; o3:=i3; o4:=i3; o5:=i3+i2; o6:=o5+i1; o8:=i1+i0; o7:=o8+i2; o9:=i0;
\end{lstlisting}
\end{minipage}
\end{proof}

\section{Rank-13 degree 4 polynomial multiplications}\label{app:Deg4}
We give here the algorithm for generic degree $4$ polynomial
multiplications over any ring of
\plinoptdata{4o4o8\_Montgomery-13-58_{L,R,P}.s{ms,lp}}.
This computes
$c_0+c_1X+c_2X^2+\ldots+c_8X^8=(a_0+a_1X+\ldots+a_4X^4)(b_0+b_1X+\ldots+b_4X^4)$,
with $13$ multiplications, $53$ additions/subtractions and $5$
scalings by $2$ or $3$.

\begin{lstlisting}[style=slp,caption={Left \& Right \SLP; 13 products;
    Post \SLP}]
l7:=a0+a1; l8:=a0-a4; l6:=a4+a3; l3:=l7-l6; x7:=a2+l7; l0:=l6+x7; x9:=l6+a2; l1:=a0-x9; l2:=x7-a4; l4:=l1+a4; l5:=l2-a0; l9:=a4; l10:=a3; l11:=a1; l12:=a0;

r7:=b0+b1; r8:=b0-b4; r6:=b4+b3; r3:=r7-r6; y7:=b2+r7; r0:=r6+y7; y9:=r6+b2; r1:=b0-y9; r2:=y7-b4; r4:=r1+b4; r5:=r2-b0; r9:=b4; r10:=b3; r11:=b1; r12:=b0;

p0:=l0*r0; p1:=l1*r1; p2:=l2*r2; p3:=l3*r3; p4:=l4*r4; p5:=l5*r5; p6:=l6*r6; p7:=l7*r7; p8:=l8*r8; p9:=l9*r9; p10:=l10*r10; p11:=l11*r11; p12:=l12*r12;

z15:=p1+p10-p8; z16:=p2+p11-p8; k2:=p6-p10; k3:=p7-p11; k4:=z15-p4; k5:=z16-p5; k9:=p12+p9; k10:=k2+z15-p3; k11:=p3+p0-k3-z16; c7:=k2-p9; c1:=k3-p12; c6:=k4+k9-k2; c2:=k5+k9-k3; c5:=k11-(k4+p9)*2-p12*3; c4:=k4+k5+k9*3+k10-k11; c3:=p0-k10-(k5+p12)*2-p9*3; c0:=p12; c8:=p9;
\end{lstlisting}

\section{Rank-11 multiplication in $\F_{3^5}$}\label{app:F243}
We give here the best known algorithm to multiply elements in
$\F_{3^5}$ in 11 multiplications, as
polynomials modulo $3$ and modulo $1-X+X^5$ of:
\plinoptdata{4o4o4_F243-11-44_{L,R,P}.s{ms,lp}}.
The associated \SLP is given thereafter.
It computes
$c_0+c_1X+\ldots+c_4X^4\equiv(a_0+a_1X+\ldots+a_4X^4)(b_0+b_1X+\ldots+b_4X^4)\mod{1-X+X^5}mod{3}$
with $11$ multiplications and $20$ additions.

\begin{lstlisting}[style=slp,caption={Left \& Right \SLP; 11 products;
    Post \SLP}]
l7:=a3-a4; x6:=a1+a0; l1:=x6-a2; l5:=a1-a4; l0:=a0+a2; l2:=a0-a3; l3:=l1-a3; l6:=a2-a3+l5; l8:=a1-l7; l9:=x6-l7; l10:=a2+l7; l4:=a4;

r7:=b3-b4; y6:=b1+b0; r1:=y6-b2; r5:=b1-b4; r0:=b0+b2; r2:=b0-b3; r3:=r1-b3; r6:=b2-b3+r5; r8:=b1-r7; r9:=y6-r7; r10:=b2+r7; r4:=b4;

p0:=l0*r0; p1:=l1*r1; p2:=l2*r2; p3:=l3*r3; p4:=l4*r4; p5:=l5*r5; p6:=l6*r6; p7:=l7*r7; p8:=l8*r8; p9:=l9*r9; p10:=l10*r10;

n1:=p7+p6; k16:=p3-n1; q5:=p5+k16; q2:=p2+k16; z7:=p1-q5; n5:=p8+q2+p1; c3:=z7-p10-p4; c1:=q2+p0-z7; n7:=p0-p10-n5; c0:=p4-p7-p3-q5+n5; c4:=p9-n7; c2:=n7+n1-q5;
\end{lstlisting}

\section{Optimal rank-8 multiplication in $\SF_{3^4}$ \#1}\label{app:s81}
Here is the best possible algorithm for a presemifield of order $81$,
as given in \plinoptdata{3o3o3_S81_8_22_{L,R,P}.s{ms,lp}}.

\begin{lstlisting}[style=slp,caption={Left \& Right \SLP; 8 products;
    Post \SLP}]
x4:=a1+a3; x5:=a0+a2; l4:=a2+x4;  y4:=b1+b2; y5:=b0+b3; r1:=b0+y4;
l5:=a3+x5; l6:=a0+x4; l7:=a1+x5;  r2:=b2+y5; r3:=b3+y4; r4:=b1+y5;
l0:=a0; l1:=a1; l2:=a2; l3:=a3;   r0:=b3; r5:=b0; r6:=b1; r7:=b2;

p0:=l0*r0; p1:=l1*r1; p2:=l2*r2; p3:=l3*r3; p4:=l4*r4; p5:=l5*r5; p6:=l6*r6; p7:=l7*r7;

n1:=p4+p3; z4:=p2+p1; c3:=p1+p0+n1; c0:=p5-p2-n1; c2:=p7+p3+z4; c1:=p6-p4+z4;
\end{lstlisting}

Here are the change-of-basis matrices that transforms the tensor $T(L,R,P)$ to that of semifield $\SF_{3^4}$ \#1 from \cite{Dempwolff:2008:s81}.
{\tiny
\[
X =
\begin{pmatrix}
2 & 2 & 1 & 1 \\
0 & 1 & 1 & 2 \\
0 & 2 & 1 & 2 \\
2 & 2 & 2 & 1
\end{pmatrix}, \quad
Y =
\begin{pmatrix}
2 & 1 & 1 & 0 \\
1 & 0 & 1 & 1 \\
0 & 1 & 1 & 1 \\
2 & 2 & 0 & 2
\end{pmatrix}, \quad
Z =
\begin{pmatrix}
2 & 0 & 1 & 2 \\
1 & 2 & 2 & 0 \\
1 & 1 & 0 & 1 \\
0 & 2 & 1 & 0
\end{pmatrix}
\]}
\section{Rank-10 multiplication in $\SF_{3^5}$ \#2}\label{app:S243}
Here is the best algorithm we found for a semifield of order $243$,
as given in \plinoptdata{4o4o4_S243_10-43__{L,R,P}.s{ms,lp}}.

\begin{lstlisting}[style=slp,caption={Left \& Right \SLP; 10 products;
    Post \SLP}]
l2:=a1+a4; x8:=a3+a4; x9:=a3-a4; l3:=a0+l2; l0:=a0-a1; l4:=a2+l3; l1:=l4-a4; l5:=a3+l2; l6:=a2-x8; l7:=l1-x8; l8:=a0+x9; l9:=a2+a1+x9;

r2:=b1+b4; y8:=b3+b4; y9:=b3-b4; r3:=b0+r2; r0:=b0-b1; r4:=b2+r3; r1:=r4-b4; r5:=b3+r2; r6:=b2-y8; r7:=r1-y8; r8:=b0+y9; r9:=b2+b1+y9;

p0:=l0*r0; p1:=l1*r1; p2:=l2*r2; p3:=l3*r3; p4:=l4*r4; p5:=l5*r5; p6:=l6*r6; p7:=l7*r7; p8:=l8*r8; p9:=l9*r9;

z5:=p9-p4; z8:=p5+p2; z7:=p5-p2; z6:=p4-p1; c0:=p8+p7-p4+z7; c4:=z7-p6+z6; c1:=p6-p3+p0+z6; c2:=z8-p3-z5; c3:=z8-p7-p1+z5;
\end{lstlisting}

\end{document}